\documentclass[sigconf,screen]{acmart}
\pdfoutput=1

\copyrightyear{2023}
\acmYear{2023}
\setcopyright{rightsretained}
\acmConference[SIGIR '23]{Proceedings of the 46th International ACM SIGIR Conference on Research and Development in Information Retrieval}{July 23--27, 2023}{Taipei, Taiwan}
\acmBooktitle{Proceedings of the 46th International ACM SIGIR Conference on Research and Development in Information Retrieval (SIGIR '23), July 23--27, 2023, Taipei, Taiwan}\acmDOI{10.1145/3539618.3591659}
\acmISBN{978-1-4503-9408-6/23/07}

\usepackage{booktabs} 
\usepackage{graphicx}
\usepackage{ifthen}

\usepackage[utf8]{inputenc} 
\usepackage[T1]{fontenc}    
\usepackage{url}            
\usepackage{booktabs}       
\usepackage{amsfonts}       
\usepackage{dsfont}

\usepackage{nicefrac}       
\usepackage{microtype}      
\usepackage{amsmath}
\usepackage{amsthm}
\usepackage{bm}
\usepackage{subcaption}
\usepackage{mathrsfs}
\usepackage{xcolor}
\usepackage{here}

\usepackage{enumitem}
\usepackage{algorithm}
\usepackage{algcompatible}
\algnewcommand\algorithmicreturn{\textbf{return }}
\algnewcommand\RETURN{\State \algorithmicreturn}%

\usepackage{multirow}
\usepackage{hyperref}
\usepackage[colorinlistoftodos,prependcaption]{todonotes}
\usepackage{natbib}
\usepackage{balance}
\usepackage{mathtools}
\newcommand{\norm}[1]{\left\lVert #1 \right\rVert}
\DeclarePairedDelimiterX{\ind}[1]{\mathbb{I}\{}{\}}{#1}

\DeclareMathOperator{\PC}{\mathcal{PC}}
\DeclareMathOperator{\conv}{conv}
\DeclareMathOperator{\nrank}{\mathsf{nperm}}

\newcommand{\calP}{\mathcal{P}}

\newcommand{\calA}{\mathcal{A}}

\newcommand{\calU}{\mathcal{U}}
\newcommand{\calV}{\mathcal{V}}

\newcommand{\calO}{\mathcal{O}}
\newcommand{\calB}{\mathcal{B}}

\newcommand{\calF}{\mathcal{F}}
\newcommand{\calQ}{\mathcal{Q}}
\newcommand{\calS}{\mathcal{S}}
\newcommand{\calH}{\mathcal{H}}

\newcommand{\R}{\mathbb{R}}
\newcommand{\w}{\mathbf{w}}

\newcommand{\x}{\mathbf{x}}

\newcommand{\p}{\mathbf{p}}

\newcommand{\z}{\mathbf{z}}

\newcommand{\W}{\mathbf{W}}

\newcommand{\bs}{\mathbf{s}}

\newcommand{\bv}{\mathbf{v}}

\newcommand{\bq}{\mathbf{q}}
\newcommand{\bh}{\mathbf{h}}
\newcommand{\be}{\mathbf{e}}
\newcommand{\bp}{\mathbf{p}}
\newcommand{\bn}{\mathbf{n}}
\newcommand{\bl}{\mathbf{l}}
\usepackage[nameinlink]{cleveref}
\crefname{theorem}{Theorem}{Theorems}
\crefname{lemma}{Lemma}{Lemmas}
\crefname{claim}{Claim}{Claims}
\crefname{remark}{Remark}{Remarks}
\crefname{observation}{Observation}{Observations}
\crefname{corollary}{Corollary}{Corollaries}
\crefname{appendix}{Appendix}{Appendices}
\crefname{section}{Section}{Sections}
\crefname{algorithm}{Algorithm}{Algorithms}
\crefname{equation}{Eq.}{Eqs.}
\crefname{figure}{Figure}{Figures}
\crefname{table}{Table}{Tables}

\newtheorem{theorem}{Theorem}[section]
\newtheorem{claim}[theorem]{Claim}
\newtheorem{observation}[theorem]{Observation}
\newtheorem{remark}[theorem]{Remark}

\usepackage{tikz}
\usetikzlibrary{positioning}
\usetikzlibrary{decorations.markings}
\usetikzlibrary{shapes.geometric}
\usetikzlibrary{arrows.meta}
\usetikzlibrary{arrows,automata,decorations.pathmorphing,fit,positioning,arrows.meta,calc,decorations.text}
\usetikzlibrary{patterns}
\usetikzlibrary{intersections}

\renewcommand{\paragraph}[1]{%
\noindentparagraph{\textbf{\textup{#1.}}}
}%

\begin{document}
\title{
Curse of ``Low'' Dimensionality in Recommender Systems
}


\author{Naoto Ohsaka}
\affiliation{
  \institution{CyberAgent, Inc.}
  \city{Tokyo} 
  \country{Japan}
}
\email{ohsaka_naoto@cyberagent.co.jp}
\orcid{0000-0001-9584-4764}
\authornote{Equal contribution: \url{https://www.aeaweb.org/journals/policies/random-author-order/search?RandomAuthorsSearch[search]=VQXAE0BZ6P_I}}

\author{Riku Togashi}
\authornotemark[1]
\affiliation{
  \institution{CyberAgent, Inc.}
  \city{Tokyo} 
  \country{Japan}
}
\email{rtogashi@acm.org}
\orcid{0000-0001-9026-0495}

\makeatletter
\gdef\@copyrightpermission{
  \begin{minipage}{0.3\columnwidth}
   \href{https://creativecommons.org/licenses/by/4.0/}{\includegraphics[width=0.90\textwidth]{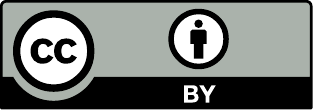}}
  \end{minipage}\hfill
  \begin{minipage}{0.7\columnwidth}
   \href{https://creativecommons.org/licenses/by/4.0/}{This work is licensed under a Creative Commons Attribution International 4.0 License.}
  \end{minipage}
  \vspace{5pt}
}
\makeatother

\begin{abstract}
Beyond accuracy, there are a variety of aspects to the quality of recommender systems, such as diversity, fairness, and robustness.
We argue that many of the prevalent problems in recommender systems are partly due to \emph{low-dimensionality} of user and item embeddings, particularly when dot-product models, such as matrix factorization, are used.

In this study, we showcase empirical evidence suggesting the necessity of sufficient dimensionality for user/item embeddings to achieve diverse, fair, and robust recommendation.
We then present theoretical analyses of 
the expressive power of dot-product models.
Our theoretical results demonstrate that
the number of possible rankings expressible under dot-product models is exponentially bounded by the dimension of item factors.
We empirically found that the low-dimensionality contributes to a popularity bias, widening the gap between the rank positions of popular and long-tail items; we also give a theoretical justification for this phenomenon.
\end{abstract}



\settopmatter{printfolios=false}
\maketitle
\section{Introduction}

Matrix factorization (MF)~\cite{hu2008collaborative} is a standard tool in recommender systems.
MF can learn compact, retrieval-efficient representation and thus is easy-to-use, particularly in large-scale applications.
On the other hand, due to the advancement of deep learning-based methods, complex nonlinear models have also been adopted to enhance recommendation quality~\cite{steck2021deep}.
However, despite the advances in model architecture, 
most models share a common structure, i.e., \emph{dot-product models}~\cite{rendle2020neural}.
Dot-product models are a class of models that estimate the preference for a user-item pair by computing a dot-product (inner product) between the user and item embeddings; MF is the simplest model in this class.
This structure is essential for large-scale applications due to its computationally efficient retrieval through vector search algorithms~\cite{jegou2010product,shrivastava2014asymmetric,malkov2018efficient}.

The \emph{dimensionality} of user/item embeddings characterizes dot-product models.
One interpretation of dimensionality is that it refers to the ranks of user/item embedding matrices that minimize the errors between true and estimated preferences.
In the extreme case where the dimensionality is one, the embedding of each user and item degenerates into a scalar.
Notice here that the rankings estimated from a feedback matrix also degenerate into two unique rankings, namely, the popularity ranking and its reverse; for ranking prediction, a one-dimensional user embedding determines only the signatures of preference scores.
Generalizing this, we arrive at a curious question: When the dimensionality in a dot-product model is low or high, what do the rankings look like?

Previous studies~\cite{zhou2008large,pilaszy2010fast,koren2022advances} have reported the effectiveness of large dimensionalities in rating prediction tasks.
\citet{rendle2021revisiting} also recently showed that high-dimensional models can achieve very high ranking accuracy under appropriate regularization.
Furthermore, successful models in top-$K$ item recommendation, such as variational autoencoder (VAE)-based models~\cite{liang2018variational}, often use large dimensions.
These observations are somewhat counterintuitive since user feedback data are typically sparse and thus may lead to overfitting under large dimensionality.
On the other hand, conventional studies of state-of-the-art methods often omit the investigation of dimensionality in models (e.g., \cite{liang2018variational,sachdeva2019sequential,he2020lightgcn,shenbin2020recvae,xie2021adversarial}).
Although exhaustive tuning of hyperparameters for model sizes (e.g., the number of hidden layers and the size of each layer) is unrealistic due to the experimental burden,
the dimensionality of user and item embeddings is rather unnoticed compared with other hyperparameters, such as learning rates and regularization weights.
This further stimulates our interest above.

In this study, we investigate the effect of dimensionality in recommender systems based on dot-product models.
We first present empirical observations from various viewpoints on recommendation quality, i.e., personalization, diversity, fairness, and robustness.
Our results reveal a hidden side effect of low-dimensionality: low-dimensional models incur a low model capacity with respect to these quality requirements even when the ranking quality seems to be maintained.
In the convention of machine learning, we can often avoid model overfitting by using low-dimensional models.
However, such models would suffer from potential long-term negative effects, namely, overfitting toward popularity bias.
Consequently, low-dimensionality leads to nondiverse, unfair recommendation results and thus insufficient data collection for producing models that can properly delineate users' individual tastes.
Furthermore, we theoretically explain the cause of the observed phenomenon, \emph{curse of low dimensionality}.
Our theoretical results---which apply to dot-product models---provide evidence
that increasing dimensionality exponentially improves the expressive power of dot-product models in terms of the number of expressible rankings, and
that we may not completely circumvent popularity bias.
Finally, we discuss possible research directions based on the lessons learned.

\begin{figure*}[h]
     \centering
    \includegraphics[keepaspectratio, width=0.85\linewidth]{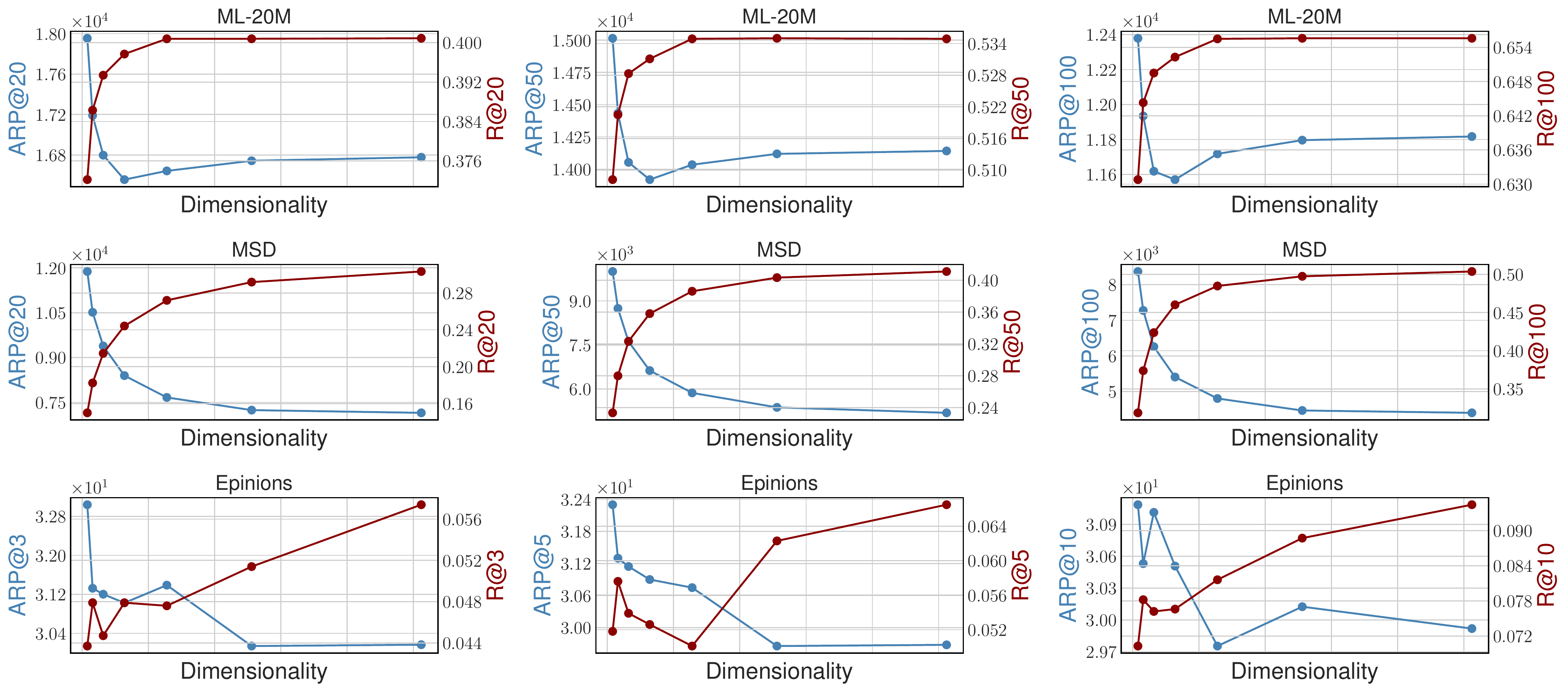}
    \caption{Effect of the dimensionality of iALS on popularity bias in recommendation results.
    }
    \label{fig:popularity-bias}
\end{figure*}

\section{Preliminary: Dot-Product Models}
In this section, we briefly describe dot-product models.
Many practical models are classified as dot-product models~\cite{rendle2020neural} (also known as two-tower models~\cite{yang2020mixed}), which estimate the preference $r_{u,v} \in \R$ of a user $u \in \calU$ for an item $v \in \calV$ by an inner-product between the embeddings of $u$ and $v$ as follows:
\begin{align*}
  \hat{r}_{u,v} = \langle \phi(u), \psi(v)\rangle,
\end{align*}
where $\phi(u) \in \R^d$ and $\psi(v) \in \R^d$ are the embeddings of $u$ and $v$, respectively.
The design of the feature mappings $\phi \colon \calU \to \R^d$ and $\psi \colon \calV \to \R^d$ depends on the overall model architecture; $\phi$ and $\psi$ can be arbitrary models, such as MF and neural networks.

(V)AEs can also be interpreted as dot-product models~\cite{liang2018variational,shenbin2020recvae}.
Most AE-based models
take partial user feedback ($|\calV|$-dimensional multi-hot vector) corresponding to one user as input and
have a fully-connected layer to make a final score prediction for $|\calV|$ items.
Denoting the $d$-dimensional intermediate representation of a partial user feedback simply by $\bq \in \R^{d}$, 
the $|\calV|$-dimensional structured prediction $\z \in \R^{|\calV|}$ can be expressed as follows:
\begin{align*}
   \z = \mathbf{f}(\W\bq + \mathbf{b}),
\end{align*}
where $\W \in \R^{|\calV| \times d}$ and $\mathbf{b} \in \R^{|\calV|}$ are the weight matrix and the bias in the fully-connected layer, respectively.
Here, $\mathbf{f} \colon \R^{|\calV|} \to \R^{|\calV|}$ is the activation function (e.g., softmax) that is order-preserving\footnote{Formally, we can say $\mathbf{f}$ is order-preserving if $\mathbf{f}$ satisfies, for any $\x \in \R^{|\calV|}$ and $\forall i,j \in \calV$ such that $i \neq j$, it holds that $\x_i > \x_j \Longrightarrow \mathbf{f}(\x)_i > \mathbf{f}(\x)_j$.}.
Given an auxiliary vector $\phi(u)=[\bq; 1] \in \R^{d+1}$ with an additional dimension of $1$ and an auxiliary matrix $\W'=[\W \,\, \mathbf{b}] \in \R^{|\calV| \times (d+1)}$ with an additional column of $\mathbf{b}$,
the predicted ranking is derived by the order of $\W'\phi(u) \in \R^{|\calV|}$ because $\mathbf{f}$ is order-preserving and does not affect the ranking prediction.
Therefore, each row of $\W'$ can be viewed as the corresponding item embedding $\psi(v)$, and thus the prediction for $u$ is also derived from the order of the dot products $\{\langle \phi(u), \psi(v) \rangle \mid v \in \calV\}$.
This point is often unregarded in empirical evaluation; for example, \citet{liang2018variational} used a large dimensionality of $600$ for VAE-based models for the Million Song Dataset and Netflix Prize; by comparison, the MF-based baseline~\cite{hu2008collaborative} employed a dimensionality of only $200$.

\section{Empirical Observation}
\label{sec:empirical}
We first present empirical observations that establish the motivation behind our theoretical analysis and discuss the possible effects of dimensionality on recommendation quality in terms of various aspects, such as personalization, diversity, item fairness, and robustness to biased feedback.

\subsection{Experimental Setting}
For our experiments, we use implicit alternating least squares (iALS)~\cite{hu2008collaborative,koren2009matrix,rendle2021revisiting}, which is widely used in practical applications and implemented in distributed frameworks, such as Apache Spark~\cite{meng2016mllib}.
Because iALS slows down in high-dimensional settings\footnote{This point is discussed in \cref{sec:discussion}}, we use a recently developed block coordinate descent solver~\cite{rendle2021ials++}.

We conduct experiments on three real-world datasets from various applications, i.e., MovieLens 20M (ML-20M)~\cite{harper2015movielens}, Million Song Dataset (MSD)~\cite{bertin2011million}, and Epinions~\cite{massa2007trust}.
To create implicit feedback datasets, we binarize explicit data by keeping observed user-item pairs with ratings larger than four for ML-20M and Epinions.
We utilize all user-item pairs for MSD.
For Epinions, we keep the users and items with more than 20 interactions as in conventional studies~\cite{abdollahpouri2017controlling,abdollahpouri2019managing}.
We also strictly follow the evaluation protocol of \citet{liang2018variational} based on a strong generalization setting.

\setcounter{figure}{2}
\begin{figure*}[t]
    \centering
    \includegraphics[keepaspectratio, width=0.85\linewidth]{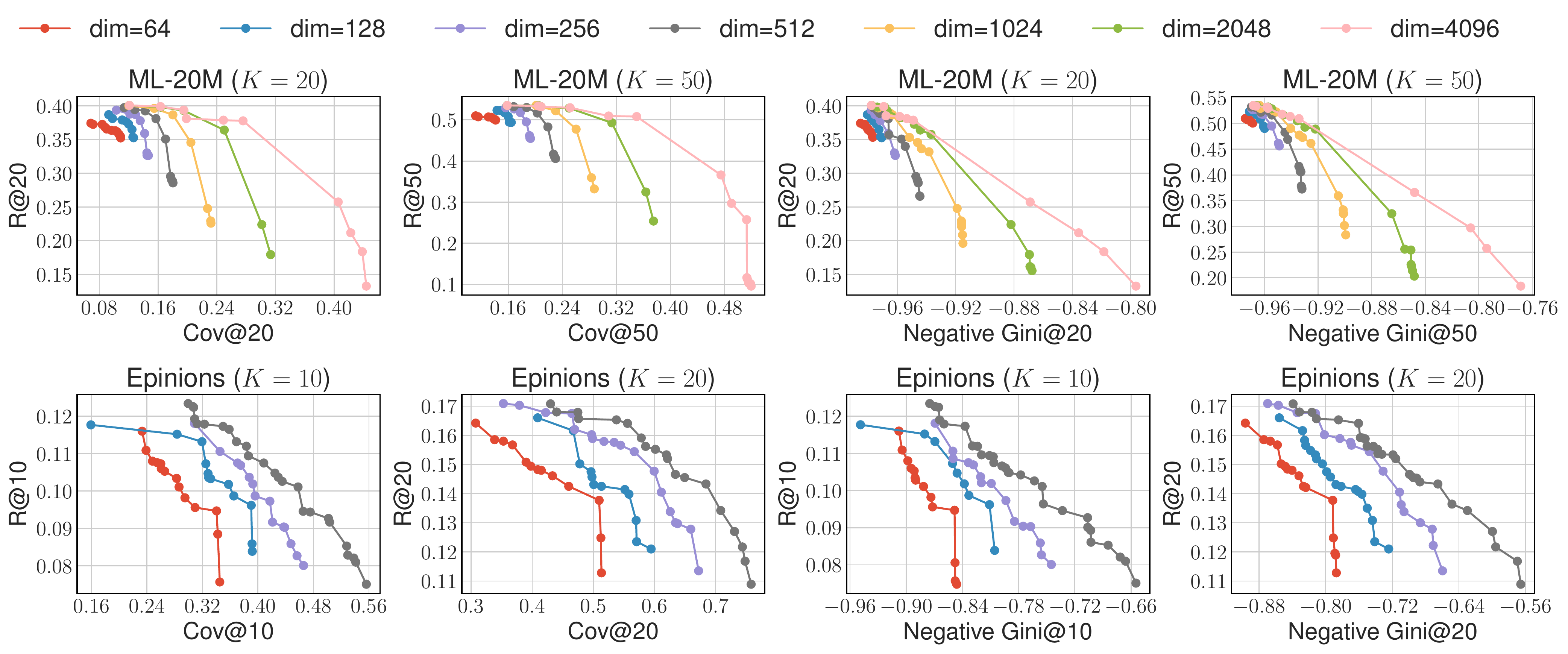}
    \caption{Effect of the dimensionality of iALS on catalog coverage and item fairness on ML-20M. Each line indicates the Pareto frontier of R@$K$ vs.~Cov@$K$ (top row) or R@$K$ vs.~Negative Gini@$K$ (bottom row) of models with fixed dimensionality.}
    \label{fig:coverage-and-fairness}
\end{figure*}

\subsection{Personalization and Popularity Bias}
\label{sec:personalization}
Personalization is the primary aim of a recommendation system, which requires a system to adapt its predictions for the users considering their individual tastes.
By contrast, the most-popular recommender is known to be an empirically strong yet \emph{anti-personalized} baseline; that is, it recommends an identical ranking of items for all users~\cite{ferrari2019we}.
Therefore, the degree of personalization in the predicted rankings may be considered as that of popularity bias~\cite{celma2008hits,yin2012challenging,abdollahpouri2017controlling,abdollahpouri2019managing}.

\cref{fig:popularity-bias} shows the personalization measure (i.e., Recall@$K$) and the average recommendation popularity (ARP@$K$)~\cite{yin2012challenging} for iALS models with various embedding sizes.
Here, ARP@$K$ is the average of item popularity (i.e., empirical frequency in the training split) for a top-$K$ list.
We tune $d \in \{64, 128, \dots, 4{,}096\}$ for ML-20M and MSD and $d \in \{8, 16, \dots, 512\}$; because the numbers of items after preprocessing are $20{,}108$ for ML-20M, $41{,}140$ for MSD, and $3{,}423$ for Epinions, we tune $d$ in different ranges to ensure the low-rank condition of MF.
It can be observed that, for all settings, the models with small dimensionalities exhibit extremely large ARP@$K$ values (see blue lines).
Particularly, at the top of the rankings, the popularity bias in the prediction is severe (shown by the leftmost figures in the top and bottom rows).
These results suggest that low-dimensional models recommend many popular items and, therefore, suffer from anti-personalization.
Furthermore, insufficient dimensionality leads to low achievable ranking quality.

On the other hand, the trends exhibited by high-dimensional models are rather different for ML-20M vs. for MSD and Epinions.
Although the ranking quality becomes saturated at a relatively low dimensionality of $1{,}024$ in ML-20M for all $K$, the quality on MSD and Epinions can be further improved with high dimensionality.
Moreover, the popularity bias in ML-20M is the lowest with $d=512$, which is not optimal in terms of ranking quality, whereas the best one in terms of ranking quality for MSD also performs the best in terms of popularity bias.
This is possibly because the popularity bias is more severe for ML-20M than for MSD; thus, reconciling high quality and low bias is difficult for ML-20M.
To confirm this,
\cref{fig:popularity-distribution} illustrates the distribution of item popularity in ML-20M, MSD, and Epinions.
The y-axis represents the normalized frequency of each item (i.e., relative item popularity) in the training dataset.
The popularity bias in ML-20M is more intense than that in MSD; there are many ``long-tail'' items in ML-20M.
Remarkably, although the results differ between ML-20M and MSD, we can draw the same conclusion: sufficient dimensionality is necessary to alleviate the existing popularity bias in ML-20M or to represent highly personalized recommendation results for MSD and Epinions.

\setcounter{figure}{1}
\begin{figure}[t]
    \centering
    \includegraphics[keepaspectratio, width=0.99\linewidth]{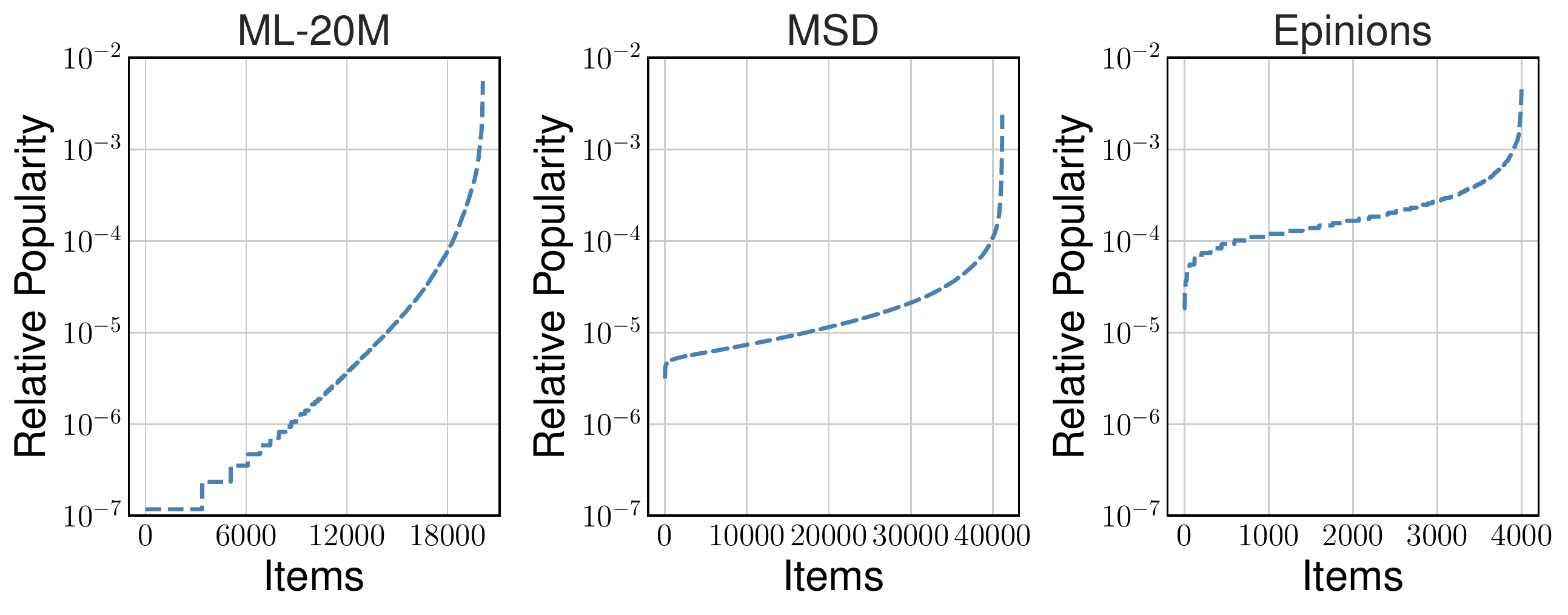}
    \caption{Distributions of relative item popularity.}
    \label{fig:popularity-distribution}
\end{figure}

\setcounter{figure}{3}
\begin{figure*}[th]
    \centering
    \includegraphics[keepaspectratio, width=0.9\linewidth]{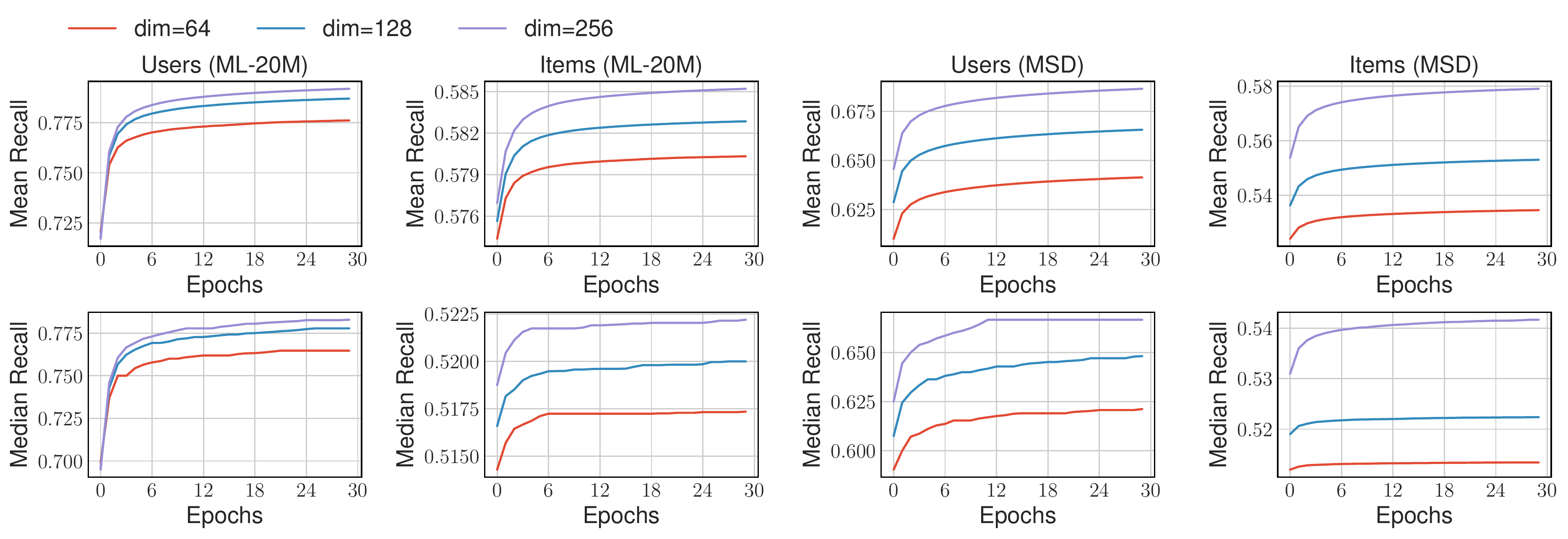}
    \caption{Effect of the dimensionality on data collection efficiency.}
    \label{fig:data-collection}
\end{figure*}

\begin{figure}[th]
    \centering
    \includegraphics[keepaspectratio, width=0.9\linewidth]{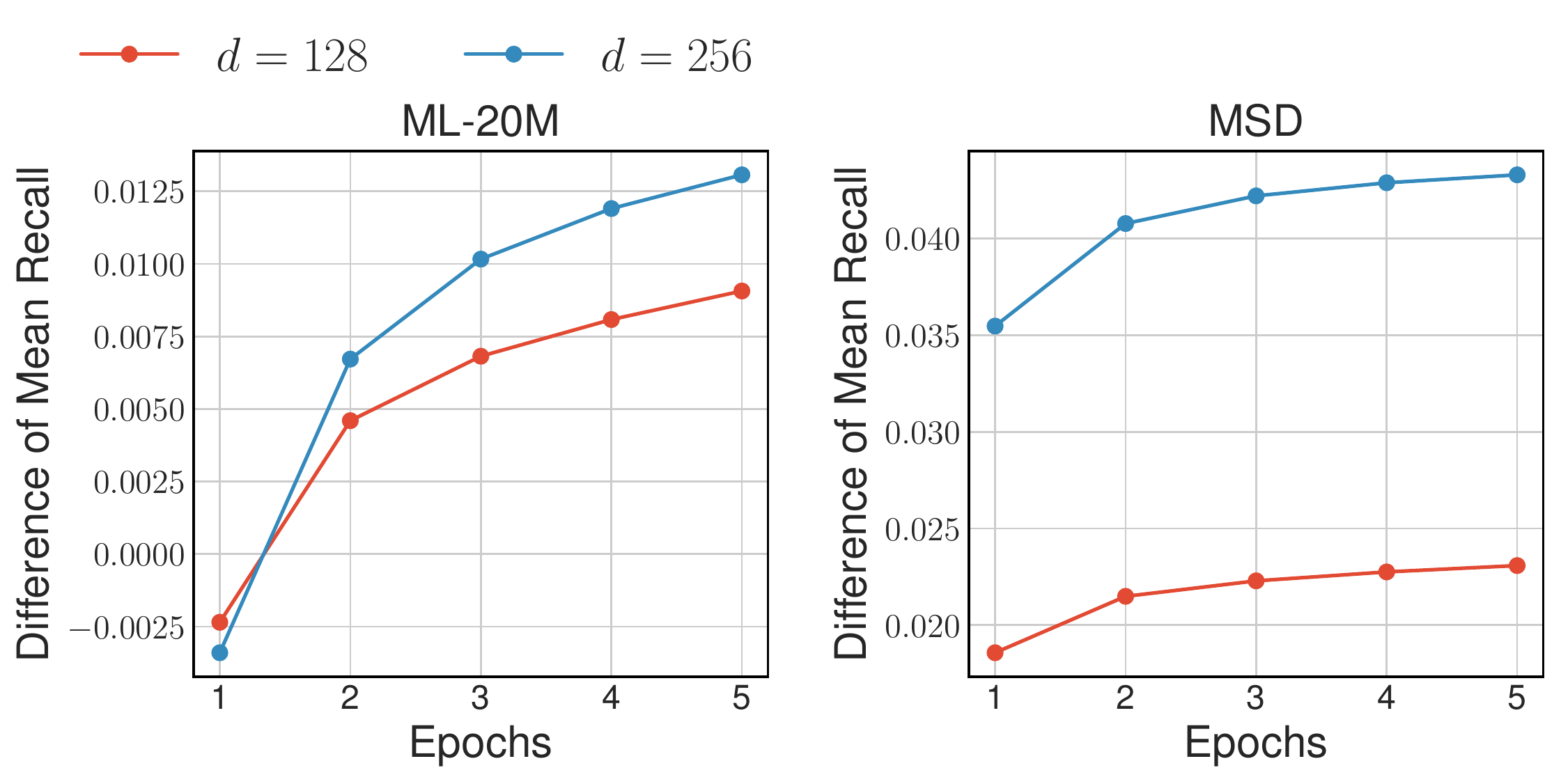}
    \caption{Improvement of models with $d=128,256$ from that with $d=64$ in terms of mean recall for users.}
    \label{fig:data-collection-recall-difference}
\end{figure}

\subsection{Diversity and Item Fairness}
\label{sec:diversity-and-fairness}
Catalog coverage is one of the concepts related to the diversity of recommendation results (i.e., aggregate diversity)~\cite{herlocker2004evaluating,adomavicius2009toward,adomavicius2011improving,adomavicius2014optimization} and refers to the proportion of items to be recommended.
This can be viewed as the capacity of an effective item catalog under a recommender system.
On the other hand, item fairness is an emerging concept similar to catalog coverage, yet it applies to different situations and requirements~\cite{burke2017multisided}.
When each item belongs to a certain user, the recommendation opportunity for the items is part of the user utility; for instance, in an online dating application where each user corresponds to an item, users (as items) can obtain more chances of success when they are recommended more frequently.

\cref{fig:coverage-and-fairness} shows the effect of dimensionality on catalog coverage and item fairness for ML-20M and Epinions.
Each curve indicates the Pareto frontier of an iALS model that corresponds to a certain dimensionality with various hyperparameter settings; 
we omitted this experiment for the large-scale MSD because of the experimental burden.
We use Coverage@$K$ (Cov@$K$) and Negative Gini@$K$ as the measures of catalog coverage and item fairness, respectively.
Cov@$K$ is the proportion of items retrieved that are in the top-$K$ results at least once~\cite{herlocker2004evaluating}, and Negative Gini@$K$ is the negative Gini index~\cite{atkinson1970measurement} of items' frequencies in the top-$K$ results\footnote{It can be defined as $\text{Negative Gini@}K= (2\norm{\mathbf{c}}_1|\calV|^2)^{-1}\sum_{(j,l) \in \calV \times \calV}|c_j - c_l|$, where $\mathbf{c} \in \mathbb{Z}_+^{|\calV|}$ is the vector of item frequencies in the top-$K$ list, and $c_j$ is the frequency of item $j \in \calV$.}.

A clear trend can be observed for all settings:
models with larger dimensions achieve higher capacities in terms of both catalog coverage and item fairness.
This result implies that low-dimensional models cannot produce diverse or fair recommendation results.
Notably, there exists a pair of models that are almost equivalent in terms of ranking quality (i.e., R@$K$) but substantially different in catalog coverage or item fairness.
This suggests a serious potential pitfall in practice.
When developers evaluate and select models based only on ranking quality, a reasonable choice may be to use a low-dimensional model due to the space cost efficiency.
However, such a model can lead to extremely nondiverse and unfair recommendation results.
Even when developers select models based on both ranking quality and diversity, 
the versatility of models is severely limited if the dimensionality parameter is tuned for a narrow range of values owing to computational resource constraints.

\subsection{Self-Biased Feedback Loop}
\label{sec:popularity-and-data-collection}
To capture the dynamic nature of user interests, a system usually retrains a model after observing data within a certain time interval.
By contrast, hyperparameters are often fixed for model retraining because hyperparameter tuning is a costly process, particularly when models are frequently updated.
Hence, the robustness of deployed models (including hyperparameters) to dynamic changes in user behavior is critical.
This may also be related to unbiased recommendation~\cite{chen2020bias} or batch learning from logged bandit feedback~\cite{swaminathan2015batch}.
Because user feedback is collected under the currently deployed system, item popularity is formed in part by the system itself.
Therefore, when a system narrows down its effective item catalog, as demonstrated in \cref{sec:diversity-and-fairness}, the data observed in the future concentrate on items that are frequently recommended by the system.
This phenomenon accelerates popularity bias in the data and further increases the number of cold-start items.

To observe the effect of dimensionality on data collection in a training and observation loop,
we repeatedly train and evaluate an iALS model with different dimensionalities on ML-20M and MSD.
Following a weak generalization setting, we first observe 50\% of the feedback for each user in the original dataset.
We then train a model using all samples in the observed dataset and predict the top-50 rankings for all users, removing the observed user-item pairs from the rankings.
Subsequently, we observe the positive pairs in the predicted rankings as an additional dataset for the next model training.
In the evaluation phase, we compute the proportion of observed samples for each user; we simply call this measure \emph{recall} for users.
Furthermore, we also compute this recall measure for items to determine the degree to which the system can collect data for each item.

\cref{fig:data-collection} shows the evolution of the recall measures for users and items.
For both ML-20M and MSD, models with higher dimensionalities achieve more efficient data collection for both users and items.
Remarkably, the difference is substantial in the data collection for items; in the figures in the second and fourth columns, a much larger efficiency gap can be observed between the high- and low-dimensional models.
Furthermore, this trend is emphasized for MSD, which has a larger item catalog than ML-20M (as shown by the figure in the fourth column of the bottom row).
The results with $d=64$ for MSD in terms of \emph{median} recall for users and items are remarkable, as shown by the red lines in the third and fourth columns of the bottom row.
The model with $d=64$ can collect data from users to some extent (third column), whereas the data come from a small proportion of items (fourth column).
Here, \cref{fig:data-collection-recall-difference} illustrates the performance gap (mean recall for users) between models with $d=64$ and $d=128,256$.
Although the gap in the first epoch is relatively small for each setting, it grows in the next few epochs owing to the different efficiency of data collection;
interestingly, in ML-20M, the models with $d=128,256$ deteriorate from that with $d=64$ only in the first epoch.
The performance gain is emphasized particularly with larger $d$ values.
These results imply that the gap between low- and high-dimensional models may become tangible, particularly in a training and observation loop, whereas the evaluation protocol in academic research simulates only the first epoch.

In summary, the results obtained in this part of the study are in good agreement with those presented in the previous sections.
Dimensionality determines the capacity in terms of various aspects of recommendation quality beyond accuracy.
As demonstrated earlier, deterioration in diversity, item fairness, and data collection ultimately affect long-term accuracy.

\subsection{Summary of Empirical Results}
Here, we summarize the empirical results obtained thus far.
We first evaluated the MF models (i.e., the simplest dot-product models) in the standard setting of item recommendation in \cref{sec:personalization}.
We observed that high-dimensional models tend to achieve high-ranking quality and low popularity bias in their predicted rankings.
To examine this phenomenon in depth, \cref{sec:diversity-and-fairness} investigated the relationship between dimensionality and diversity/fairness.
We obtained a clear trend showing that low-dimensionality limits the versatility of models in terms of diversity and item fairness rather than ranking quality; even when the other hyperparameters are tuned, the achievable performance in terms of diversity/fairness is in a narrow range.
In \cref{sec:popularity-and-data-collection}, we further investigated the effect of such a limited model space on data collection and long-term accuracy.
We simulated a practical situation wherein a system re-trains models by using incrementally 
observed data under its own recommendation policy.
The results suggest that the data collected under low-dimensional models are severely biased by the model itself and thus impede accuracy improvement in the future.
In the following section, we shall dive into the mechanism of this ``curse of low-dimensionality'' phenomenon.

\begin{figure}[t]
    \centering
    \scalebox{0.7}{\begin{tikzpicture}
    \tikzset{node/.style={circle, fill=black, font=\LARGE, text centered, inner sep=0, outer sep=0, minimum size=0.25cm}};
    \tikzset{dotted-edge/.style={ultra thick, densely dashed}};
    \tikzset{bold-edge/.style={very thick}};
    \tikzset{arrow-edge/.style={ultra thick, ->, >=Stealth}};


    \coordinate(a) at (-0.8,1.8);
    \coordinate(b) at (2.8,0.8);
    \coordinate(c) at (2.8,-2.0);
    \coordinate(d) at (-3.3,-2.0);
    \coordinate(O) at (0,0);
    \draw(a) node[node]{};
    \draw(b) node[node]{};
    \draw(c) node[node]{};
    \draw(d) node[node]{};
    \draw(O) node[node]{};
    \node[above=0.1cm of a] {\LARGE $\bv_1$};
    \node[right=0.1cm of b] {\LARGE $\bv_2$};
    \node[below=0.1cm of c] {\LARGE $\bv_3$};
    \node[left=0.1cm of d] {\LARGE $\bv_4$};
    \node at (-0.5,-0.2) {\LARGE $O$};

    \foreach \v/\w in {a/b, b/c, c/d, d/a, b/d, a/c}
        \draw[dotted-edge] (\v)--(\w);
    
    \foreach \v / \w in {a/b, b/c, c/d, d/a, b/d, a/c}
        \draw[bold-edge] ($(O)!-3cm!($(\v)!(O)!(\w)$)$)--($(O)!3cm!($(\v)!(O)!(\w)$)$);
    
    \foreach \v/\w in {a/b, b/c, c/d, d/a, b/d, a/c}
        \draw[semithick]
            let \p1 = ($(\v)!(O)!(\w)$)
            in
            ($(\p1)!0.2cm!(\v)$)--($(\p1)!0.2cm!(\v)!-0.2cm!90:(\v)$)--($(\p1)!-0.2cm!(O)$);

    \foreach \v/\w in {d/a}
        \draw[arrow-edge, color=blue!65!black]
            let \p1 = ($(\v)!(O)!(\w)$)
            in
            ($(\p1)!-0.5cm!(O)$)--($(\p1)!1cm!(\w)!0.5cm!90:(\w)$);

    \foreach \v/\w in {d/a}
        \draw[arrow-edge, color=red!65!black]
            let \p1 = ($(\v)!(O)!(\w)$)
            in
            ($(\p1)!-0.5cm!(O)$)--($(\p1)!1cm!(\v)!-0.5cm!90:(\v)$);

    \node[anchor=west, rotate=55.5, color=blue!65!black] at (-2.4,1.6) {$\langle \bq, \bv_1 \rangle > \langle \bq, \bv_4 \rangle$};
    \node[anchor=east, rotate=55.5, color=red!65!black] at (-2.4,1.6) {$\langle \bq, \bv_4 \rangle > \langle \bq, \bv_1 \rangle$};

\end{tikzpicture}}
    \caption{
        Illustration of the proof idea of \cref{thm:theory:rank-informal},
        which characterizes representable rankings of size $n$ by
        regions in hyperplane arrangements.
        There are four item vectors $\bv_1, \bv_2, \bv_3, \bv_4$ on $\R^2$.
        Each dashed line connects a pair of the four vectors;
        each bold line is orthogonal to some dashed line and intersects the origin.
        These bold lines generate twenty regions in total, each corresponds to a distinct ranking.
    }
    \label{fig:theory:rank-n}
\end{figure}
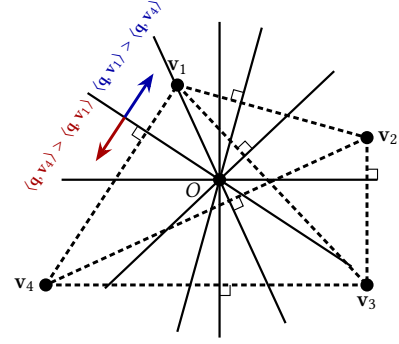

\section{Theoretical Investigation}
\label{sec:theory}
We present (a summary of) theoretical analyses on the expressive power of dot-product models
to support the empirical results provided in \cref{sec:empirical},
whose formal statements and proofs are deferred to \cref{app:theory}.
Our results are twofold and highly specific to dot-product models:
\paragraph{Bounding the Number of Representable Rankings (cf.~\cref{subsec:theory:bound})}
First, we investigate how many different rankings may be expressed over a fixed list of item vectors.
Our hypothesis is that
we suffer from popularity bias and/or cannot achieve diversity and fairness satisfactorily
owing to low expressive power.
In particular, 
we are interested in the number of representable rankings
parameterized by
the number of items $n$,
dimensionality $d$, and
size of rankings $K$.

Slightly formally, we say that a ranking $\pi$ of size $K$ is \emph{representable}
over $n$ item vectors $\bv_1, \ldots, \bv_n$ in $\R^d$
if
there exists a \emph{query vector} $\bq \in \R^d$ (e.g., a user embedding) such that
$\pi$ is consistent with
a ranking uniquely obtained by arranging $n$ items in descending order of $\langle \bq, \bv_i \rangle$.
We devise upper and lower bounds on the number of representable rankings,
informally stated as:
\begin{theorem}[informal; see \cref{thm:theory:rank-n,thm:theory:rank-d}]
\label{thm:theory:rank-informal}
The following holds:
\begin{itemize}[leftmargin=*]
\item \textbf{Upper bound}:
For every $n$ item vectors in $\R^d$,
the number of representable rankings of size $K$ over them is at most
$n^{\min\{K, 2d\}}$.
\item \textbf{Lower bound}:
There exist $n$ item vectors in $\R^d$ such that
the number of representable rankings of size $d$ over them 
is $n^{\Theta(d)}$.
\end{itemize}
\end{theorem}
Our upper and lower bounds imply that 
the maximum possible number of representable rankings of size $K$
is essentially $n^{\Theta(\min\{K,d\})}$,
offering the following insight:
\emph{increasing the dimensionality $d$ ``exponentially'' improves the expressive power of dot-product models}.
\cref{fig:theory:rank-n,fig:theory:rank-d}
illustrate the proof idea of our upper and lower bounds,
which characterize representable rankings by hyperplane arrangement and facets of a polyhedron,
respectively.

\begin{figure}[t]
    \centering
    \scalebox{0.7}{\begin{tikzpicture}
    \tikzset{node/.style={circle, fill=black, font=\LARGE, text centered, inner sep=0, outer sep=0, minimum size=0.25cm}};
    \tikzset{dotted-edge/.style={ultra thick, densely dashed}};
    \tikzset{bold-edge/.style={very thick}};
    \tikzset{arrow-edge/.style={ultra thick, ->, >=Stealth}};


    \coordinate(a) at (-0.8,1.8);
    \coordinate(b) at (2.8,0.8);
    \coordinate(c) at (2.8,-2.0);
    \coordinate(d) at (-3.3,-2.0);
    \coordinate(e) at (-0.9,-0.5);
    \coordinate(f) at (1.4,-0.8);
    \coordinate(O) at (0,0);
    \draw(a) node[node]{};
    \draw(b) node[node]{};
    \draw(c) node[node]{};
    \draw(d) node[node]{};
    \draw(e) node[node]{};
    \draw(f) node[node]{};
    \draw(O) node[node]{};
    \node[above=0.1cm of a] {\LARGE $\bv_1$};
    \node[right=0.1cm of b] {\LARGE $\bv_2$};
    \node[below=0.1cm of c] {\LARGE $\bv_3$};
    \node[left=0.1cm of d] {\LARGE $\bv_4$};
    \node[below=0.1cm of e] {\LARGE $\bv_5$};
    \node[below=0.1cm of f] {\LARGE $\bv_6$};
    \node[left=0.1cm of O] {\LARGE $O$};
    \node[text centered] at (0,-1.7) {\textcolor{black}{\LARGE $\calP = \conv(\{\bv_1, \bv_2, \bv_3, \bv_4, \bv_5, \bv_6\})$}};
    
    \foreach \v/\w in {a/b, b/c, c/d, d/a}
        \draw[dotted-edge] (\v)--(\w);
    

    \foreach \v/\w in {a/b, b/c, c/d, d/a}
        \draw[arrow-edge] ($(\v)!(O)!(\w)$)--($($(\v)!(O)!(\w)$)!-1cm!(O)$);
    
    \foreach \v/\w in {a/b, b/c, c/d, d/a}
        \draw[semithick]
            let \p1 = ($(\v)!(O)!(\w)$)
            in
            ($(\p1)!0.2cm!(\v)$)--($(\p1)!0.2cm!(\v)!-0.2cm!90:(\v)$)--($(\p1)!-0.2cm!(O)$);

\end{tikzpicture}}
    \caption{
        Illustration of the proof idea of \cref{thm:theory:rank-informal},
        which partially characterizes representable rankings of size $d$
        by facets of a polyhedron.
        There are six item vectors
        $\bv_1,\bv_2,\bv_3,\bv_4,\bv_5,\bv_6$ on $\R^2$,
        of which convex hull $\calP$ has $\bv_1,\bv_2,\bv_3,\bv_4$ as its vertices;
        each dashed segment represents a facet of $\calP$.
        For each facet,
        there exists a ranking of size $2$ dominated by any two item vectors on the facet.
    }
    \label{fig:theory:rank-d}
\end{figure}
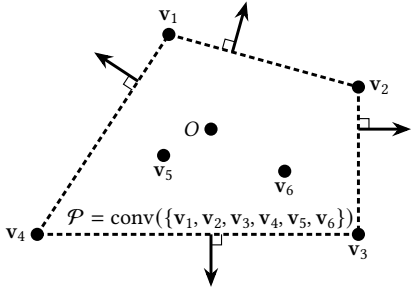

\paragraph{Mechanism behind Popularity Bias (cf.~\cref{subsec:theory:popularity})}
We then study the mechanism behind popularity bias and its effect on the space of representable rankings.
Consider first that there exists a pair of sets consisting of \emph{popular} and \emph{long-tail} items.
Usually, many users prefer popular items than long-tail ones,
which turns out to restrict the space from which user embeddings are chosen;
we can thus easily infer that
some similar-to-popular items rank higher than long-tail items as well.

Slightly formally,
given a pair of sets,
$P = \{\bp_1, \ldots, \bp_{|P|}\}$ and $L = \{\bl_1, \ldots, \bl_{|L|}\}$,
consisting of popular and long-tail items,
we force a query vector $\bq$ (e.g., a user embedding)
to always ensure that items of $P$ rank higher than items of $L$.
Under this setting,
we establish a structural characterization of such ``near-popular'' items that are ranked higher than $L$,
informally stated as:
\begin{theorem}[informal; see \cref{lem:theory:singleton,thm:theory:multi}]
\label{lem:theory:cone-informal}
Suppose that a query vector $\bq$ ranks all items of $P$ higher than all items of $L$.
Then,
if an item vector $\bs$ is included in a particular convex cone
that contains the convex hull of $P$,
then $\bs$ ranks higher than every item of $L$
(i.e., $\bs$ is popular).

Moreover, this cone becomes bigger as more popular or long-tail items are added
(i.e., $|P|$ or $|L|$ is increased).
\end{theorem}

\begin{example}
\cref{fig:theory:multi} shows an illustration of \cref{lem:theory:cone-informal}.
We are given three popular items $P = \{\bp_1,\bp_2,\bp_3\}$ and
two long-tail items $\bl_1,\bl_2$.
Given that a query vector $\bq$ ranks
the three popular items higher than $\bl_1$ \emph{only},
$\bq$ ranks another item $\bs$ higher than $\bl_1$ whenever
$\bs$ is in $\calS(P,\{\bl_1\})$,
which is a convex cone denoted by northeast blue lines.
Similarly, if $\bq$ ranks $P$ higher than $\bl_2$,
$\bs$ in $\calS(P,\{\bl_2\})$, denoted by northwest red lines,
ranks higher than $\bl_2$.

Consider now the case that
popular items of $P$ rank higher than \emph{both} $\bl_1$ and $\bl_2$.
Then, an item $\bs$ ranks higher than $\bl_1$ and $\bl_2$ if
$\bs$ is included in $\calS(P,\{\bl_1,\bl_2\})$,
which is a convex cone denoted by two arrows.
This convex cone is larger than $\calS(P,\{\bl_1\})$ and $\calS(P,\{\bl_2\})$.
\end{example}

The above theorem suggests that
the existence of a small number of popular and long-tail items
makes other near-popular items superior to long-tail ones,
reducing the number of representable rankings, and thus
we may not completely avoid popularity bias.

In conclusion, our theoretical results justify the empirical observations:
The limited catalog coverage with low dimensionality in \cref{sec:diversity-and-fairness}
is due to lack of the number of representable rankings as in \cref{thm:theory:rank-informal},
and 
the popularity bias (a large value of ARP@$K$) observed in \cref{sec:personalization}
is partially explained by \cref{lem:theory:cone-informal} and the discussion in \cref{subsec:theory:popularity}.
Furthermore,
as our theoretical methodology only assumes that
the underlying architecture follows dot-product models,
the counter-intuitive phenomena of low-dimensionality
would be applied to not only for iALS but also for any dot-product models.

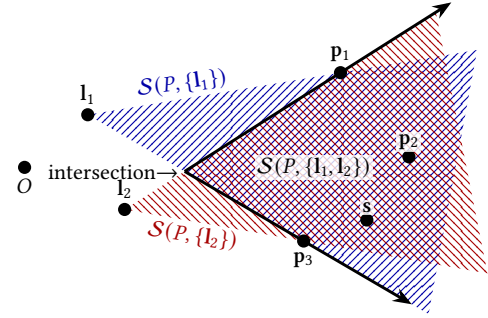
\begin{figure}[t]
    \centering
    \scalebox{0.7}{\begin{tikzpicture}
    \tikzset{node/.style={circle, fill=black, font=\LARGE, text centered, inner sep=0, outer sep=0, minimum size=0.25cm}};
    \tikzset{dotted-edge/.style={ultra thick, densely dashed}};
    \tikzset{bold-edge/.style={very thick}};
    \tikzset{arrow-edge/.style={ultra thick, ->, >=Stealth}};


    \coordinate(u1) at (1.2,1.0);
    \coordinate(u2) at (1.9,-0.8);
    \coordinate(p1) at (6.0,1.8);
    \coordinate(p2) at (7.3,0.2);
    \coordinate(p3) at (5.3,-1.4);
    \coordinate(ss) at (6.5,-1.0);
    \coordinate(O) at (0,0);

    \path[pattern=north east lines, pattern color=blue!65!black] (u1)--($(u1)!7.5cm!(p1)$)--($(u1)!7.5cm!(p3)$)--cycle;
    \path[pattern=north west lines, pattern color=red!65!black] (u2)--($(u2)!7cm!(p1)$)--($(u2)!7cm!(p3)$)--cycle;

    \draw(u1) node[node]{};
    \draw(u2) node[node]{};
    \draw(p1) node[node]{};
    \draw(p2) node[node]{};
    \draw(p3) node[node]{};
    \draw(ss) node[node]{};
    \draw(O) node[node]{};
    \node[above=0.1cm of u1] {\LARGE $\bl_1$};
    \node[above=0.1cm of u2] {\LARGE $\bl_2$};
    \node[above=0.1cm of p1] {\LARGE $\bp_1$};
    \node[above=0.1cm of p2, inner sep=0.04cm, fill=white, opacity=0.8, text opacity=1] {\LARGE $\bp_2$};
    \node[below=0.1cm of p3] {\LARGE $\bp_3$};
    \node[above=0.1cm of ss, inner sep=0.04cm, fill=white, opacity=0.8, text opacity=1] {\LARGE $\bs$};
    \node[below=0.1cm of O] {\LARGE $O$};
    \node[text centered, rotate=8] at (3,1.6) {\textcolor{blue!65!black}{\LARGE $\calS(P, \{\bl_1\})$}};
    \node[text centered, rotate=-8] at (3.2,-1.3) {\textcolor{red!65!black}{\LARGE $\calS(P, \{\bl_2\})$}};
    \node[text centered, inner sep=0.04cm, fill=white, opacity=0.8, text opacity=1] at (5.5,-0.0) {\LARGE $\calS(P, \{\bl_1,\bl_2\})$};
    
    \path[name path=u1p3] (u1)--(p3);
    \path[name path=u2p1] (u2)--(p1);
    \path[name intersections={of=u1p3 and u2p1, by={c}}];
    \draw[arrow-edge] (c)--($(c)!6cm!(p1)$);
    \draw[arrow-edge] (c)--($(c)!5cm!(p3)$);
    \node[text centered, left=0cm of c] {\LARGE intersection$\rightarrow$};

\end{tikzpicture}}
    \caption{
        Illustration of \cref{lem:theory:cone-informal}.
        Let $P = \{\bp_1,\bp_2,\bp_3\}$ be a set of three popular items and
        $L = \{\bl_1,\bl_2\}$ be a set of two long-tail items.
        Any item in $\calS(P,\{\bl_1\})$ (resp.~$\calS(P,\{\bl_2\})$)
        denoted northeast blue (resp.~northwest red) lines,
        ranks higher than $\bl_1$ (resp.~$\bl_2$) whenever all popular items of $P$ rank higher than $\bl_1$ (resp.~$\bl_2$).
        The intersection of $\calS(P,\{\bl_1\})$ and $\calS(P,\{\bl_2\})$ is crosshatched.
        $\calS(P,\{\bl_1,\bl_2\})$ is a convex cone defined by two arrows, which includes $\calS(P,\{\bl_1\}) \cap \calS(P,\{\bl_2\})$.
        Any item in $\calS(P,\{\bl_1,\bl_2\})$ ranks higher than both $\bl_1$ and $\bl_2$
        whenever all items of $P$ rank higher than all items of $L$.
    }
    \label{fig:theory:multi}
\end{figure}

\section{Discussion and Future Direction}
\label{sec:discussion}
\paragraph{Efficient Solvers for High Dimensionality}
High-dimensional models are often computationally costly.
Even in the most efficient methods based on MF,
the optimization involves cubic dependency on $d$ and thus does not scale well for high-dimensional models.
Motivated by this, previous studies have explored efficient methods for high-dimensional models~\cite{he2016fast,bayer2017generic,rendle2021ials++}.
Because the traditional ALS solver for MF is in the class of block coordinate descent (CD), conventional methods can be derived by designing the choice of blocks~\cite{rendle2021ials++}.
It is hence interesting to design block CD methods for high-dimensional models considering efficiency on modern computer architectures (e.g., CPU with SIMD, GPU, and TPU~\cite{mehta2021alx}).
Developing solvers for more complex models, such as factorization machines~\cite{rendle2012factorization,bayer2017generic}, is useful to leverage side information.
Since deep learning-based models are also employed, 
efficient solvers for such non-linear models~\cite{yang2020mixed} will be beneficial for enhancing the ranking quality.
The extension of conventional algorithms to stochastic optimization, which uses a subset of data samples in a single update, is also important for using memory-intensive complex models and large-scale data.
To reduce the memory footprint of embeddings, sparse representation may also be promising~\cite{ning2011slim}.
As sparsity constraints introduce another difficulty in optimization, 
some techniques, such as the alternating direction method of multipliers (ADMM)~\cite{boyd2011distributed}, would be required.
ADMM is a recent prevalent approach to enable parallel and scalable optimization under constraints~\cite{yu2014distributed,steck2020admm,steck2021negative,togashi2022fair}.
Improving the serving cost of high-dimensional models is also essential by using maximum inner-product search (MIPS)~\cite{jegou2010product,shrivastava2014asymmetric,malkov2018efficient}.

\paragraph{Efficient Methods for Diversity and Item Fairness}
Our empirical results suggest that diversity and item fairness may also be beneficial for efficient data collection and long-term accuracy.
We can also infer from our results that directly optimizing diversity and item fairness is a promising approach to enhancing long-term recommendation quality.
From a practical point of view, however, diverse and fair item recommendation is often computationally costly because of combinatorial explosions in large-scale settings.
Hence, it is an important direction to develop efficient diversity-aware recommendation~\cite{yao2016tweet,chen2017improving,chen2018fast,wilhelm2018practical,warlop2019tensorized}.
In the same spirit as dot-product models, 
efficient sampling techniques based on MIPS are essential for real-time retrieval~\cite{han2020map,hirata2022solving}.
On the other hand, fairness-aware item recommendation is a challenging task in terms of its computational efficiency.
Because fairness requires controlling item exposure while considering all rankings for users, 
it is relatively complex in both optimization and prediction phases compared to fairness-agnostic top-$K$ ranking problems.
There exist a variety of approaches to implementing fair recommender systems based on constraint optimization~\cite{celis2018ranking,singh2018fairness,patro2020fairrec,memarrast2021fairness,togashi2022fair}.
Although we focus on offline collaborative filtering,
online recommendation methods based on bandit algorithms~\cite{li2016collaborative,korda2016distributed,mahadik2020fast,patil2021achieving,wang2021fairness,jeunen2021top} are also promising for directly integrating accuracy optimization and data collection.

\paragraph{Further Theoretical Analysis on Dot-product Models}
Our theoretical results might help us gain a deeper understanding of the expressive power of dot-product models.
Besides the open problems described in \cref{sec:theory}, there is still much room for further investigation of its expressive power.
We may consider a fine-grained analysis of representable $K$-permutations for $K \not\in \{d,n\}$; e.g.,
to establish the exact upper or lower bounds of $\nrank_K$.
One major limitation of \cref{thm:theory:rank-informal} is that they do not promise that 
\emph{every} $n^{\Theta(d)}$ ranking is representable (over some $n$ $d$-dimensional item vectors); i.e.,
we cannot rule out the existence of a small set of rankings that cannot be expressed under dot-product models.
Thus, a possible direction is to analyze the representability of an input set $\Pi$ of rankings:
Can we construct $n$ item vectors over which 
each ranking of $\Pi$ is representable?
This question may be thought of as an \emph{inverse problem} to that discussed in \cref{subsec:theory:bound}.

\section{Conclusion}
In this paper, we presented empirical results that reveal the necessity of sufficient dimensionality in dot-product models.
Our results suggest that low-dimensionality leads to overfitting to existing popularity bias and further amplifies the bias in future data.
This phenomenon, referred to as \emph{curse of low dimensionality}, partly causes modern problems in recommender systems, such as personalization, diversity, fairness, and robustness to biased feedback.
Then, we theoretically discussed the phenomenon from the viewpoint of \emph{representable rankings}, which are the rankings that a model can represent.
We showed the bound on the number of representable rankings for $d$-dimensional models.
This result suggests that low-dimensionality leads to an exponentially small number of representable rankings.
We also explained the effect of popular items on representable rankings.
Finally, we established a structural characterization of near-popular items, suggesting a mechanism behind popularity bias under dot-product models.

\appendix

\section{Formal Statements and Proofs}
\label{app:theory}

\subsection{Number of Representable Rankings}
\label{subsec:theory:bound}
We devise lower and upper bounds on the number of representable rankings parameterized by
the number of items,
dimensionality, and 
size of rankings.
Hereafter, we identify the set $\calV$ of $n$ items with $[n] \triangleq \{1,2,\ldots, n\}$.
For nonnegative integers $n \in \mathbb{N}$ and $K \in [n]$,
let $\mathfrak{S}_n$ denote the set of all permutations over $[n]$, and
let $\mathfrak{S}_n^K$ denote the set of all $K$-permutations of $[n]$
(also known as partial permutations); e.g.,
$\mathfrak{S}_3^2 = \{(1,2),\allowbreak(1,3),\allowbreak(2,1),\allowbreak(2,3),\allowbreak(3,1),\allowbreak(3,2)\}$.
Note that
$\mathfrak{S}_n^n = \mathfrak{S}_n$,
$|\mathfrak{S}_n| = n!$, and
$|\mathfrak{S}_n^K| = n^{\underline{K}}$,
where $n^{\underline{K}} \triangleq \frac{n!}{(n-K)!}$ denotes the falling factorial.
We consider each $K$-permutation $\pi \in \mathfrak{S}_n^K$ as
a ranked list of $K$ items
such that item $\pi(i)$ ranks in the $i$-th place.

We now define the representability of $K$-permutations.
For $n$ items,
let $\bv_1, \ldots, \bv_n$ be vectors in $\R^d$ representing their embeddings.
Without much loss of generality, we assume that 
they are in \emph{general position}; i.e., no $d+1$ vectors lie on a hyperplane.
Given a \emph{query vector} $\bq \in \R^d$ (e.g., a user embedding),
we generally produce a ranking obtained by
arranging $n$ items in descending order of $\langle \bq, \bv_i \rangle$.
We thus say that
$\bq$ over $\bv_1, \ldots, \bv_n$ \emph{represents}
a $K$-permutation $\pi \in \mathfrak{S}_n^K$ if
$\langle \bq, \bv_{\pi(i)} \rangle > \langle \bq, \bv_{\pi(j)} \rangle$ for all $1 \leq i < j \leq K$
and $\langle \bq, \bv_{\pi(K)} \rangle > \langle \bq, \bv_{\pi(i)} \rangle$ for all $K+1 \leq i \leq n$,
and that
$\pi$ is \emph{representable} if such $\bq$ exists.
Here, we emphasize that ``ties'' are not allowed;
i.e., $\bq$ does not represent $\pi$ whenever
$\langle \bq, \bv_i \rangle = \langle \bq, \bv_j \rangle$ for some $i \neq j$.
We let $\nrank_K(\bv_1, \ldots, \bv_n)$ be
the number of representable $K$-partial permutations over $\bv_1, \ldots, \bv_n$;
namely,
\begin{align*}
    \nrank_K(\bv_1, \ldots, \bv_n) \triangleq
    |\{
        \pi \in \mathfrak{S}_n^K :
        \pi \text{ is representable over } \bv_i\text{'s}
    \}|.
\end{align*}
By definition,
$\nrank_K(\bv_1, \ldots, \bv_n) \leq n^{\underline{K}} \leq n^K$.
Our first result is an upper bound on $\nrank_n$.

\begin{theorem}
\label{thm:theory:rank-n}
For any $n$ vectors $\bv_1, \ldots, \bv_n$ in $\R^d$ in general position,
$
    \nrank_n(\bv_1, \ldots, \bv_n) \leq n^{2d}.
$
In particular,
$\nrank_n(\bv_1, \ldots, \bv_n) \leq n^{\min\{K, 2d\}}$
for every $K \in [n]$.
\end{theorem}
The proof uses a characterization of representable permutations by hyperplane arrangements,
which is illustrated in \cref{fig:theory:rank-n}.

Subsequently, we provide a lower bound on $\nrank_d$ in terms of the number of facets of a polyhedron.
Let $\calP \triangleq \conv(\{\bv_1, \ldots, \bv_n\}) \subset \R^d$
be a \emph{convex hull} of $\bv_1, \ldots, \bv_n$;
every vertex of $\calP$ corresponds to some $\bv_i$.
For a vector $\mathbf{a} \in \R^d$ and scalar $z \in \R$,
a linear inequality $\langle \mathbf{a}, \mathbf{x} \rangle \leq b$ is said to be \emph{valid}
if $\langle \mathbf{a}, \mathbf{x} \rangle \leq b$ for all $\mathbf{x} \in \calP$.
A subset $\calF$ of $\calP$ is called a \emph{face} of $\calP$ if
$
    \calF = \calP \cup \{ \mathbf{x} : \langle \mathbf{a},\mathbf{x} \rangle = b \}
$
for some valid linear inequality $\langle \mathbf{a},\mathbf{x} \rangle \leq b$.
In particular, $(d-1)$-dimensional faces are called \emph{facets}.
Every facet includes exactly $d$ vertices of $\calP$ by definition (whenever $n$ vectors are in general position).
Our second result is as follows.

\begin{theorem}
\label{thm:theory:rank-d}
For any $n$ vectors $\bv_1, \ldots, \bv_n$ in $\R^d$ in general position,
$\nrank_d(\bv_1, \ldots, \bv_n)$ is at least the number of facets of $\calP$.
Moreover, there exist $n$ vectors $\bv_1, \ldots, \bv_n$ in $\R^d$ such that
\begin{align*}
    \nrank_d(\bv_1, \ldots, \bv_n)
    \geq {n \choose \left\lceil d/2 \right\rceil}
    \geq \left(\frac{n}{d/2}\right)^{d/2}.
\end{align*}
\end{theorem}
By \cref{thm:theory:rank-n,thm:theory:rank-d},
the maximum possible number of representable $K$-permutations over $n$ vectors in $\R^d$ is
$n^{\Theta(\min\{K,d\})}$ for all $K \in [n]$ and $d = \calO(1)$.

What remains to be done is the proof of \cref{thm:theory:rank-n,thm:theory:rank-d}.

\begin{proof}[Proof of \cref{thm:theory:rank-n}]
For each $1 \leq i < j \leq n$,
we introduce a \emph{pairwise preference} $\delta_{i,j} \in \{\pm 1\}$ between $i$ and $j$.
We wish for a query vector $\bq \in \R^d$ to ensure that
item $i$ (resp.~$j$) ranks higher than item $j$ (resp.~$i$) if $\delta_{i,j}$ is $+1$ (resp.~$-1$).
This requirement is equivalent to
$\langle \bq, \bv_i-\bv_j \rangle \cdot \delta_{i,j} > 0$.
Thus,
if the following system of linear inequalities is feasible,
any of its solutions $\bq$ represents a unique permutation consistent with $\delta_{i,j}$'s:\footnote{
    Note that for \cref{eq:theory:lineq} to be feasible,
    $\delta_{i,j}$'s should be \emph{transitive}; i.e.,
    $\delta_{i,j}=+1$ and $\delta_{j,k}=+1$ imply $\delta_{i,k}=+1$.
}
\begin{align}
\label{eq:theory:lineq}
    \langle \bq, \bv_i-\bv_j \rangle \cdot \delta_{i,j} > 0 \quad \text{ for all } 1 \leq i < j \leq n.
\end{align}
On one hand, for any representable permutation $\pi \in \mathfrak{S}_n$,
\cref{eq:theory:lineq} defined by
\begin{align*}
    \delta_{i,j} = 
    \begin{cases}
    +1 & \text{if } i \text{ ranks higher than } j, \\
    -1 & \text{if } j \text{ ranks higher than } i, \\
    \end{cases}
\end{align*}
must be feasible.
On the other hand,
two different sets of pairwise preferences derive distinct permutations (if they exist).
Hence, $\nrank_n(\bv_1, \ldots, \bv_n)$ is equal to 
the number of $\delta_{i,j}$'s for which \cref{eq:theory:lineq} has a solution.

Observe now that \cref{eq:theory:lineq} is feasible
\emph{if and only if}
the intersection of $\calH_{i,j}$'s for all $1 \leq i < j \leq n$ is nonempty,
where $\calH_{i,j}$ is an open half-space defined as
\begin{align*}
    \calH_{i,j}
    \triangleq \{ \bq \in \R^d : \langle \bq, \bv_i-\bv_j \rangle \cdot \delta_{i,j} > 0 \}.
\end{align*}
Because $\calH_{i,j}$ is obtained via the division of $\R^d$ by
a unique hyperplane that is orthogonal to $\bv_i - \bv_j$ and intersects the origin,
the number of $\delta_{i,j}$'s for which \cref{eq:theory:lineq} is feasible
is equal to 
\emph{the number of regions generated by hyperplane arrangement},
which has been investigated in geometric combinatorics.
By \cite{ho2006number,winder1966partitions},
the number of regions generated by
${n \choose 2}$ $(d-1)$-dimensional hyperplanes that have a common point
is at most
\begin{align}
\label{eq:theory:regions}
    2\sum_{0 \leq i \leq d-1} {{n \choose 2}-1 \choose i}
    \leq 2 \cdot {n \choose 2}^{d-1}
    \leq n^{2d},
\end{align}
thereby completing the proof.
\end{proof}

\begin{remark}
\cref{fig:theory:rank-n} illustrates
the equivalence between representable permutations and regions of hyperplane arrangements.
There are four vectors $\bv_1, \bv_2, \bv_3, \bv_4$ on $\R^2$.
Each dashed line connects a pair of the four vectors;
each bold line is orthogonal to some dashed line and intersects the origin.
These hyperplanes generate twenty regions, each of which expresses a distinct permutation.
The number ``$12$'' is tight because
the left-hand side of \cref{eq:theory:regions} is
$2 \left({{4 \choose 2}-1 \choose 0} + {{4 \choose 2}-1 \choose 1}\right) = 12$.
On the other hand, a ranking $(\bv_1, \bv_3, \bv_2, \bv_4)$
is not representable.
\end{remark}

Before going into the proof of \cref{thm:theory:rank-d},
we prove the following,
which partially characterizes representable $d$-permutations and
is illustrated in \cref{fig:theory:rank-d}:
\begin{claim}
\label{clm:theory:facet}
For any set $I \in {[n] \choose d}$ of $d$ items,
if there exists a facet $\calF$ of $\calP$ including $\bv_i$ for every $i \in I$,
there is a query vector $\bq \in \R^d$ that represents a $d$-permutation consisting only of $I$.
\end{claim}
\begin{proof}[Proof of \cref{clm:theory:facet}]
By the definition of facets,
there exist $\mathbf{a} \in \R^d$ and $b \in \R$ such that
$\langle \bv_i, \mathbf{a} \rangle = b$ for all $i \in I$ while 
$\langle \bv_i, \mathbf{a} \rangle < b$ for all $i \not\in I$.
Thus, letting $\bq \triangleq \mathbf{a}$ completes the proof.
\end{proof}

\begin{remark}
\cref{fig:theory:rank-d} gives an illustration of \cref{clm:theory:facet}.
There are six vectors
$\bv_1,\bv_2,\bv_3,\bv_4,\bv_5,\bv_6$ on $\R^2$.
The convex hull $\calP$ of them has $\bv_1,\bv_2,\bv_3,\bv_4$ as its vertices;
each dashed segment represents a facet of $\calP$.
Taking an outward normal vector to a facet,
denoted a bold arrow, as a query vector,
we obtain a $2$-permutation dominated by two vectors on the facet.
\end{remark}

\begin{proof}[Proof of \cref{thm:theory:rank-d}]
Because of \cref{clm:theory:facet},
the number of representable $d$-permutation is at least the number of facets of $\calP$.
By the \emph{upper bound theorem} due to {McMullen}~\cite{motzkin1957comonotone,mcmullen1970maximum},
the maximum possible number of facets of
a convex polytope consisting of $n$ vertices in $\R^d$ is 
$
    {n \choose \left\lceil d/2 \right\rceil}.
$
This upper bound can be achieved when $\calP$ is a cyclic polytope,\footnote{
The \emph{cyclic polytope} is defined as
$\conv(\{ \mathbf{\alpha}_d(t_1), \ldots, \mathbf{\alpha}_d(t_n) \}) $ for distinct $t_i$'s,
where $\mathbf{\alpha}_d(t) \triangleq (t, t^2, \ldots, t^d)$ is the moment curve.
}
thus completing the proof.
\end{proof}

\begin{example}
\label{ex:theory:facet}
Finally, we show by a simple example that
facets do not fully characterize representable $d$-permutations; i.e.,
\cref{clm:theory:facet} is not tight.
Four item vectors $\bv_1, \bv_2, \bv_3, \bv_4$ in $\R^2$ are defined as:
\begin{align*}
    \bv_1 = (+1, 0), 
    \bv_2 = (-1, 0), 
    \bv_3 = (0, +1), 
    \bv_4 = (0, +\tfrac{1}{2}).
\end{align*}
The convex hull $\calP$ is clearly a triangle formed by $\{\bv_1,\bv_2,\bv_3\}$.
By \cref{clm:theory:facet},
for each facet of $\calP$,
there is a representable $2$-permutation dominated by the vectors on the facet.
However, the query vector $\bq = (0,1)$ produces
a $2$-permutation dominated by $\bv_3$ and $\bv_4$, even though $\bv_4$ lies on no facet of $\calP$.
We leave the complete characterization of representable $d$-permutations as an open problem.
\end{example}

\subsection{Mechanism behind Popularity Bias}
\label{subsec:theory:popularity}
Subsequently, we study the mechanism behind popularity bias and its effect on the space of representable rankings.
Suppose we wish a query vector $\bq \in \R^d$ to ensure that
the embedding of $|P|$ popular items denoted
$P \triangleq \{\bp_1, \ldots, \bp_{|P|}\} \subset \R^d$
ranks higher than
the embedding of $|L|$ long-tail items denoted
$L \triangleq \{\bl_1, \ldots, \bl_{|L|}\} \subset \R^d$; that is,
it must hold that
$
    \langle \bq, \bp_i \rangle > \langle \bq, \bl_j \rangle \text{ for all } i\in[|P|], j\in[|L|].
$
Let $\calQ(P,L)$ be the closure of the set of all such vectors;\footnote{
We take the closure for the sake of analysis.
} namely,
\begin{align*}
    \calQ(P,L) \triangleq
    \{
        \bq \in \R^d :
        \langle \bq, \bp_i \rangle \geq \langle \bq, \bl_j \rangle, \; \forall i\in[|P|], j \in [|L|]
    \}.
\end{align*}
We can easily decide whether $\calQ(P,L)$ is empty or not as follows.
\begin{observation}
\label{obs:theory:divide}
$\calQ(P,L) \neq \emptyset$
if and only if
there exists a hyperplane that divides $P$ and $L$.
\end{observation}

Now, it is considered that the query vectors are chosen only from $\calQ(P,L)$.
What type of item vector would always rank higher than $L$?
We define $\calS(P,L)$ as the closure of the set of vectors in $\R^d$ that rank higher than all vectors of $L$
provided that $\bq \in \calQ(P,L)$; namely,
\begin{align*}
    \calS(P,L) \triangleq
    \{
        \bs \in \R^d :
        \langle \bq, \bs \rangle \geq \langle \bq, \bl_j \rangle, \; \forall j \in [|P|], \bq \in \calQ(P,U)
    \}.
\end{align*}
First, it is observed that $\calS(P,L)$ is convex;
in fact,
for any vector $\bs = \mu \bs_1 + (1-\mu) \bs_2$ such that
$\bs_1, \bs_2 \in \calS(P,L)$ and $\mu \in [0,1]$,
we have 
\begin{align*}
    \langle \bq, \bs \rangle
    & = \mu \langle \bq, \bs_1 \rangle + (1-\mu) \langle \bq, \bs_2 \rangle \\
    & \geq \mu \max_{j \in [|P|]}\langle \bq, \bl_j \rangle + (1-\mu) \max_{j \in [|P|]}\langle \bq, \bl_j \rangle
    = \max_{j \in [|P|]}\langle \bq, \bl_j \rangle
\end{align*}
whenever $\bq \in \calQ(P,L)$.
In particular, any convex combination of $P$ ranks higher than $L$ if $\bq \in \calQ(P,L)$.

To characterize the structure of $\calS(P,L)$,
we further introduce some definitions and notations.
For two subsets $\calA$ and $\calB$ of $\R^d$,
the \emph{Minkowski sum}, denoted $\calA \oplus \calB$, is defined as
the set of all sums of a vector from $\calA$ and a vector from $\calB$; namely,
$
\calA \oplus \calB \triangleq \{ \mathbf{a} + \mathbf{b} : \mathbf{a} \in \calA, \mathbf{b} \in \calB \}.
$
The \emph{half-line} for a vector $\bv \in \R^d$ is defined as
$
    \overrightarrow{\bv} \triangleq \{ \lambda \bv : \lambda \geq 0 \}.
$
The \emph{polyhedral cone} generated by $k$ vectors $\bv_1, \ldots, \bv_k \in \R^d$ is defined as
the conical hull of $\bv_1, \ldots, \bv_k$; namely,
\begin{align*}
    \PC(\bv_1, \ldots, \bv_k)
    \triangleq \Bigl\{ \sum_{i \in [k]}\lambda_i \bv_i : \forall \lambda_i \geq 0 \Bigr\}
    = \bigoplus_{i \in [k]} \overrightarrow{\bv_i}.
\end{align*}
Note that $\PC(\bv_1, \ldots, \bv_k)$ is a convex cone.
For $k+1$ vectors $\bv_1, \ldots, \bv_k, \allowbreak \be \in \R^d$,
the polyhedral cone generated by $\bv_1, \ldots, \bv_k$ with $\be$ as an extreme point is defined as
\begin{align*}
    \PC(\bv_1, \ldots, \bv_k; \be)
    & \triangleq \be \oplus \PC(\bv_1, \ldots, \bv_k).
\end{align*}
By the Minkowski--Weyl theorem \cite{minkowski1896geometrie,weyl1935elementare},
any polyhedral cone can be represented as
the intersection of a finite number of half-spaces with $\be$ on their boundary.

We first establish a complete characterization of $\calS(P,L)$
when $L$ is a singleton consisting of a vector $\bl \in \R^d$;
see also \cref{fig:theory:multi}.
\begin{theorem}
\label{lem:theory:singleton}
It holds that
$
    \calS(P,\{\bl\}) = \PC(\{\bp_1-\bl, \ldots, \bp_k-\bl\}; \bl).
$
\end{theorem}

We further show that $\calS(P,L)$ becomes very large as $P$ or $L$ grows.
\begin{theorem}
\label{thm:theory:multi}
It holds that
\begin{align}
    \calS(P,L)
    & \supseteq
    \Bigl( \bigcap_{j \in [|L|]} \calS(P,\{\bl_j\}) \Bigr) \oplus
    \bigoplus_{i \in [|P|], j \in [|L|]} \overrightarrow{\bp_i - \bl_j} \label{eq:theory:multi} \\
    & \supseteq \conv(\{\bp_1, \ldots, \bp_{|P|}\}) \oplus \bigoplus_{i \in [|P|], j \in [|L|]} \overrightarrow{\bp_i - \bl_j}. \label{eq:theory:multi-conv}
\end{align}
\end{theorem}

\begin{remark}
\cref{fig:theory:multi} gives an example of \cref{lem:theory:singleton,thm:theory:multi}.
Let $P = \{\bp_1,\bp_2,\bp_3\}$ and $L = \{\bl_1,\bl_2\}$.
$\calS(P,\{\bl_1\})$ and $\calS(P,\{\bl_2\})$
are denoted by northeast (blue) lines and northwest (red) lines, respectively.
The intersection of $\calS(P,\{\bl_1\})$ and $\calS(P,\{\bl_2\})$ is crosshatched.
$\calS(P,L)$ is a convex cone defined by two arrows, which includes $\calS(P,\{\bl_1\}) \cap \calS(P,\{\bl_2\})$.
\end{remark}
$\calS(P,L)$ particularly contains $\conv(\{\bp_1, \ldots, \bp_{|P|}\})$, and
a polyhedral cone defined by $\bigoplus_{i \in [|P|], j \in [|L|]} \overrightarrow{\bp_i - \bl_j}$
becomes bigger as either
the number of popular items or the number of long-tail items is increased.
In particular, \cref{eq:theory:multi-conv} is 
\emph{monotone} in $P$ and $L$ (with respect to the inclusion property).
We leave as open whether the inclusion in \cref{eq:theory:multi} is strict or not.

Here, we discuss the dimensionality of $\calS(P,L)$.
If $\dim(\calS(P,L))$ is less than $d$,
a vector randomly taken from $\R^d$ lies on $\calS(P,L)$ with probability $0$;
thus, we would not be concerned about popularity bias.
Unfortunately, $\dim(\calS(P,L))$ can be bounded only by 
\begin{align*}
    \dim\Bigl( \bigcap_{j} \calS(P,\{\bl_j\}) \Bigr)
    + \dim\Bigl(\bigoplus_{i,j} \overrightarrow{\bp_i - \bl_j}\Bigr)
    \leq (|P|+1)|L|,
\end{align*}
if the inclusion in \cref{eq:theory:multi} is an equality.
Thus, we cannot inherently avoid popularity bias
under dot-product models.

However, a characterization of $\calS(P,L)$ by a cone-like structure implies
the effect of increasing the dimensionality on popularity bias reduction.
Suppose that $\calS(P,L)$ is a $d$-dimensional right circular cone of half-angle $\theta$.
Then, the event in which \emph{a random item vector uniformly chosen from $\R^d$ belongs to $\calS(P,L)$}
occurs with a probability equal to
the area ratio of
a spherical cap of polar angle $\theta$ in the $(d-1)$-dimensional unit sphere
to
the $(d-1)$-dimensional unit sphere.
By \cite{li2011concise}, this ratio is
$
    \frac{\frac{1}{2} A_{d-1} \cdot I_{\sin^2 \theta}\left(\frac{d-1}{2}, \frac{1}{2}\right)}{A_{d-1}}
    = \frac{1}{2} I_{\sin^2 \theta}\left(\frac{d-1}{2}, \frac{1}{2}\right),
$
where $A_{d-1}$ is the area of the $(d-1)$-dimensional unit sphere and
$I_x(a,b)$ is the regularized incomplete beta function.
By using the asymptotic expansion of $I_x(a,b)$ for fixed $b$ and $x$ \cite{temme1996special} and

an approximation of the beta function $\mathrm{B}(a,b) \approx \Gamma(b) \cdot a^{-b}$,
we have
$
    I_{\sin^2 \theta}\left(\frac{d-1}{2}, \frac{1}{2}\right)
    \approx \frac{(\sin^2 \theta)^{\frac{d-1}{2}} \cdot (\cos^2 \theta)^{-\frac{1}{2}}}{\frac{d-1}{2} \cdot \Gamma(\frac{1}{2})(\frac{d-1}{2})^{-\frac{1}{2}}}
    = \Theta\left(\frac{(\sin \theta)^{d-1}}{\sqrt{d-1}}\right);
$
i.e., the probability of this event decays \emph{exponentially} in $d$.

In the remainder of this section, we prove \cref{lem:theory:singleton,thm:theory:multi}.

\begin{proof}[Proof of \cref{lem:theory:singleton}]
We first show
$\PC(\{\bp_1-\bl, \ldots, \bp_{|P|}-\bl\}; \bl) \subseteq \calS(P,\{\bl\})$.
Let $\bs$ be any vector in $\PC(\{\bp_1-\bl, \ldots, \bp_{|P|}-\bl\}; \bl)$;
hence, $\bs = \bl + \sum_{i \in [|P|]}\lambda_i (\bp_i - \bl)$
for some $\lambda_i \geq 0$.
We then have that for every $\bq \in \calQ(P,\{\bl\})$,
\begin{align*}
    \langle \bq, \bs \rangle
    = \langle \bq, \bl \rangle +
        \sum_{i \in [|P|]} \lambda_i \cdot \Bigl(
        \underbrace{\langle \bq, \bp_i \rangle - \langle \bq, \bl \rangle}_{\geq 0}
    \Bigr)
    \geq \langle \bq, \bl \rangle.
\end{align*}
Thus, $\bs \in \calS(P,\{\bl\})$.

We then show (the contraposition of) that
$\calS(P,\{\bl\}) \subseteq \PC(\{\bp_1-\bl, \ldots, \bp_{|P|}-\bl\}; \bl)$.
Let $\bs$ be a vector not in $\PC(\{\bp_1-\bl, \ldots, \bp_{|P|}-\bl\}; \bl)$.
By the Minkowski--Weyl theorem,
there must be a half-space $\calH$ that has $\bl$ on its boundary,
contains $\PC(\{\bp_1-\bl, \ldots, \bp_{|P|}-\bl\}; \bl)$, and does not include $\bs$.
Note that the boundary of $\calH$ forms a supporting hyperplane.
Letting $\bn$ be an outward normal vector to $\calH$,
we can express $\bs$ as $\bs = \bh + \lambda \bn$, where
$\lambda > 0$ and $\bh$ is a vector lying on the boundary of $\calH$.
By the definition of $\calH$ and $\bn$,
it holds that $\langle -\bn, \bl\rangle \leq \langle -\bn, \bp_i \rangle$ for all $i \in [|P|]$;
i.e., $-\bn \in \calQ(P,\{\bl\})$.
Thinking of $-\bn$ as a query vector, we obtain
\begin{align*}
    \langle -\bn, \bs \rangle
    = \langle -\bn, \bh \rangle + \lambda \cdot \langle -\bn, \bn \rangle
    < \langle -\bn, \bh \rangle
    = \langle -\bn, \bl \rangle,
\end{align*}
where the last equality is because
both $\bh$ and $\bl$ lie on the boundary of $\calH$,
implying that $\bs \not\in \calS(P,\{\bl\})$.
This completes the proof.
\end{proof}

\begin{proof}[Proof of \cref{thm:theory:multi}]
First, we show that
$\calS(P,L)$ contains $\bigcap_{j \in [|L|]} \calS(P,\{\bl_j\})$.
Observing that 
$\calQ(P,L) = \bigcap_{j \in [|L|]} \calQ(P,\{\bl_j\})$,
we derive the following:
\begin{align*}
    \bigcap_{j \in [|L|]} \calS(P,\{\bl_j\})
    & = \bigcap_{j \in [|L|]} \{
        \bs : \langle \bq, \bs \rangle \geq \langle \bq, \bl_j \rangle, \; \forall \bq \in \calQ(P,\{\bl_j\})
    \} \\
    & \subseteq \bigcap_{j \in [|L|]} \{
        \bs : \langle \bq, \bs \rangle \geq \langle \bq, \bl_j \rangle, \; \forall \bq \in \calQ(P,L)
    \} \\
    & = \{
        \bs : \langle \bq, \bs \rangle \geq \langle \bq, \bl_j \rangle,
            \; \forall j \in [|L|], \forall \bq \in \calQ(P,L)
    \},
\end{align*}
which is equal to $\calS(P,L)$.
Now, let $\bs$ be in the right-hand side of \cref{eq:theory:multi}; i.e.,
it holds that $\bs = \bs' + \sum_{i,j} \lambda_{i,j}(\bp_i-\bl_j) $ for
some $\bs' \in \bigcap_{j \in [|L|]} \calS(P,\{\bl_j\})$ and $\lambda_{i,j} \geq 0$.
For any $\bq \in \calQ(P,L)$, we have
\begin{align*}
    \langle \bq, \bs \rangle
    = \langle \bq, \bs' \rangle + \sum_{i,j} \lambda_{i,j} \cdot
        \Bigl( \underbrace{\langle \bq,\bp_i \rangle - \langle \bq,\bl_j \rangle}_{\geq 0} \Bigr)
    \geq \langle \bq, \bs' \rangle \geq \max_{j \in [|L|]} \{\langle \bq, \bl_j \rangle\}, 
\end{align*}
implying $\bs \in \calS(P,L)$, which concludes the proof.
\end{proof}

\clearpage
\balance
\bibliographystyle{ACM-Reference-Format}
\bibliography{head}
\end{document}